\documentclass[final,1p,times]{elsarticle}

\usepackage{amssymb}
\usepackage{amsmath}
\usepackage{amsthm}
\usepackage{mathrsfs}
\usepackage{mathtools}
\usepackage[mathscr]{eucal}
\newtheorem{proposition}{Proposition}

\begin{document}

\begin{frontmatter}

\title{A heavy-tail arctan-based mixture model for modelling and measuring actuarial risk }
\author{Pankaj Kumar\corref{cor1}}
\ead{kumar.131@iitj.ac.in} 

\author{Vivek Vijay}
\ead{vivek@iitj.ac.in}
\affiliation{organization={Indian Institute of Technology Jodhpur},
            addressline={Karwar}, 
            city={Jodhpur},
            postcode={332030}, 
            state={Rajasthan},
            country={India}}

\begin{abstract}
Heavy-tailed probability distributions are extremely useful and play a crucial role in modeling different types of financial data sets. This study presents a two-pronged methodology. First, a mixture probability distribution is created by combining Gaussian and Rayleigh distributions using the arctangent transformation, aimed at producing heavier-tailed features and enhancing alignment with real market data. Some statistical properties of the proposed model are also discussed. Furthermore, essential actuarial risk evaluation instruments, such as value-at-risk (VaR), tail value-at-risk (TVaR) and tail variance (TV) are employed for efficient risk management practices. Lastly, an application is provided using an insurance dataset to demonstrate the applicability of the proposed model. The proposed model demonstrates superior fitting performance compared to current baseline distributions, showcasing its practical value in financial risk evaluation. The combination of Gaussian and Rayleigh distributions through arctangent transformation is particularly successful in representing extreme market behaviour and tail dependencies that are frequently found in real-world financial data.  
\end{abstract}


\end{frontmatter}

\section{Introduction}
\label{sec1}
Heavy-tailed probability distributions have received considerable focus in various research domains, such as actuarial science, biomedical research, engineering applications, risk assessment, and economic analysis \cite{alkhairy2021arctan}. Current research initiatives are centered on creating innovative techniques to establish new types of heavy-tailed probability distributions that provide improved descriptive power and greater modeling adaptability. In the realms of finance and actuarial science, a primary challenge is the accurate forecasting of significant monetary losses \cite{miljkovic2016modeling}. If organizations fail to adequately estimate these potential losses, they may encounter severe operational issues, including the threat of bankruptcy and ineffective premium pricing strategies. To tackle these challenges and enhance the precision of loss predictions, actuaries are continually looking for more versatile heavy-tailed distributions that can effectively represent the complexities found in real-world financial situations  \cite{tung2021arcsine, he2020arcsine, atchade2025modeling}. 

Standard statistical distributions often lack the necessary flexibility to accurately model datasets with heavy tails. To address this issue, numerous standard heavy-tailed distributions to effectively model heavy-tailed data are discussed in literature, including the log-normal model \cite{Singer2013}, Generalized Pareto distribution (GPD) family \cite{Arnold2015}, Weibull \cite{Lai2011}, and Gumbel \cite{Gomez2019}. However, these standard heavy-tailed models also have limitations in capturing fat-tailed data. Progress in the development of heavy-tailed distributions has been achieved through various innovative methodological techniques, such as variable transformation methods, quantile analysis, distribution composition, augmenting parameters in existing distributions, distribution compounding, and finite mixture models \cite{lee2013methods}. For instance, various distributions, such as the Arctan-X family of heavy-tailed models \cite{alkhairy2021arctan}, new distributions based on the Tan-G family \cite{souza2021tan}, exponential T-X transformation-based distributions \cite{Ahmad2021}, hyperbolic sine function based probability distributions \cite{Ahmad2023}, and new Topp Leone Kumaraswamy Marshall-Olkin generated families of distributions \cite{Atchade2024}, have been introduced for modeling the heavy tailed data. A novel family of probability distributions utilizing the sine function was proposed in \cite{Souzay2019}. A family of probability distributions leveraging sine and cosine functions was ropoed in \cite{Chesneau2019}. Conversely, other researchers have created new families of probability distributions employing inverse trigonometric functions, such as a new family of probability distributions using the arcsine and arccosine function \cite{Chesneau2021}. In \cite{Chaudhary2021}, the authors proposed the Arctan Lomax distribution. A new versatile log-logistic distribution based on the tangent function was introduced in \cite{Muse2021}. 

Financial and actuarial data commonly exhibit unique traits: they are usually heavy-tailed, possess unimodal characteristics, demonstrate right-skewed distributions, and include positive values \cite{Alamer, Ahmad}. The complex nature of today's financial data necessitates the creation of advanced probability distribution families that can address these intricate patterns. The introduction of innovative statistical models has significantly improved the efficacy of quantitative analytical methods. Although there has been considerable progress in developing new statistical frameworks, researchers still face basic challenges while trying to fit real-world data into traditional statistical framework. Given the critical role of heavy-tailed distributions in actuarial work, there is a strong impetus within the actuarial community to create new flexible distribution models. This domain is experiencing rapid advancement, with researchers utilizing various innovative strategies to build more refined models. 

\subsection{ Contribution and Outline}
In this article, we presented a new heavy-tailed distribution that is essentially developed using the arctan function, incorporating a mixture of Gaussian and Rayleigh densities. The primary benefit of this new mixture distribution based on trigonometric functions is its capacity to create new distributions while utilizing a minimal number of parameters, offering greater flexibility for modeling heavy-tailed data. Additionally, the probability density function (pdf) and the closed form of the cumulative distribution function (cdf) for the proposed distribution are straightforward to express. We additionally explore some statistical properties, including the quantile function, moments, skewness and kurtosis, survival function, and hazard function and also discuss the maximum likelihood estimation method for parameter estimation for the proposed distribution. This study also examines essential risk assessment metrics, specifically value at risk (VaR), tail value at risk (TVaR) and tail variance (TV) in the context of assessing risk within a mixture distribution framework that integrates Gaussian and Rayleigh distributions through the arctangent function. The study demonstrates the practical implementation of this method using insurance data. The rationale for selecting the mixture distribution for risk prediction lies not only in the absence of an explicit formula for calculating risk measures under this distribution but also in its propensity to produce fatter tails, which better align with actual data observations. From a financial perspective, the choice of this distribution is justified by its advantages over both Gaussian and Rayleigh distributions.

The following sections are structured as follows: Section 2 introduces the mixture distribution framework based on the arc-tangent function. In, Section 3 explain the some statistical properties of proposed distribution.  Section 4 formulates the risk assessment metrics value at risk (VaR), tail value at risk (TVaR) and tail variance (TV) for the proposed distribution. Section 5, we gave the parameter estimation method for proposed distribution. Following that, Section 6 explores the application of the developed distribution alongside recognized risk evaluation methods for insurance data.

\section{Mixture distribution based on Arctan}
In this section, we examine the combined mixture distribution formed by utilizing Gaussian \cite{karim2011rayleigh} and Rayleigh \cite{beckmann1964rayleigh} distributions through arctan transformation. One of the advantages of the Arctan-X family is that, its probability density function (PDF) can be easily articulated, alongside a manageable and closed-form cumulative distribution function (CDF). Let, $X$  be a random variable from the Arctan-$X$ distribution family. The associated CDF and PDF are respectively defined as
\begin{equation}
G_{\text{arctan}}(x; \boldsymbol{\delta}) = \frac{4}{\pi} \arctan(H(x; \boldsymbol{\delta})), \quad x \in \mathbb{R},
\end{equation}
and 
\begin{equation}
g_{\text{arctan}}(x; \boldsymbol{\delta}) = \frac{4}{\pi} \frac{h(x;\boldsymbol{\delta)}}{1+H(x; \boldsymbol{\delta})^2}, \quad x \in \mathbb{R},
\end{equation}
where $\boldsymbol{\delta} \in \mathbb{R}$ signifies the parameter vector, and $H(x; \boldsymbol{\delta})$ represents the CDF of underlying random variable. In this context, $H(x; \boldsymbol{\delta})$ describes the mixture distribution, which is characterized as a weighted linear combination of multiple probability distributions \cite{ahmed2021managing}. Such distributional frameworks prove particularly useful when modeling specific datasets with unique characteristics. $H(x; \boldsymbol{\delta})$ consisting of Gaussian and Rayleigh components, where $X \thicksim N(\omega,\eta^2)$ with probability density function (PDF) is
\begin{equation}
f_{X}(x \mid \eta) = \frac{1}{\eta\sqrt{2\pi}} \exp\left(-\frac{(x - \omega)^2}{2\eta^2}\right), \quad x,\omega \in \mathbb{R},
\end{equation}
and $\eta$ follows a Rayleigh distribution with scale parameter $\psi>0$, and its PDF is defined as
\begin{equation}
f_{\eta}(\eta) = \frac{\eta}{\psi^2} \exp\left(-\frac{\eta^2}{2\psi^2}\right), \quad x\geq 0.
\end{equation}
Then the mixture PDF of $H(x; \boldsymbol{\delta})$ is obtained by integrating the conditional PDF of $X$ given $\eta$ over all possible values of $\eta$, weighted by the PDF of $\eta$. This is an application of the law of total probability for continuous mixtures:
\begin{equation}
h(x) = \int_{0}^{\infty} f_{X \mid \eta}(x \mid \eta) \cdot f_{\eta}(\eta)  \, d\eta
\label{eq:total_prob}
\end{equation}

Substituting equations (3) and (4) into (5), we get
\[
h(x) = \int_{0}^{\infty} \left[ \frac{1}{\eta\sqrt{2\pi}} \exp\left(-\frac{(x - \omega)^2}{2\eta^2}\right) \right] \cdot \left[ \frac{\eta}{\psi^2} \exp\left(-\frac{\eta^2}{2\psi^2}\right) \right]  \, d\eta
\]
that is
\begin{equation}
h(x) = \frac{1}{\sqrt{2\pi} \psi^2} \int_{0}^{\infty} \exp\left(-\frac{(x - \omega)^2}{2\eta^2} - \frac{\eta^2}{2\psi^2} \right)  \, d\eta
\end{equation}
To evaluate the integral, we define constants $a$ and $b$ to simplify the exponent
\[
a = \frac{(x - \omega)^2}{2}, \quad b = \frac{1}{2\psi^2}
\]
Thus, the integral becomes
\begin{equation}
I = \int_{0}^{\infty} \exp\left(-\frac{a}{\eta^2} - b\eta^2 \right)  \, d\eta
\label{eq:integral_I}
\end{equation}

\[
 = \frac{1}{2}\sqrt{\frac{\pi}{b}} \exp(-2\sqrt{ab})
\]
Substituting back the values for $a$ and $b$, we get
\[
I = \frac{1}{2}\sqrt{\frac{\pi}{\frac{1}{2\psi^2}}} \exp\left(-2\sqrt{ \frac{(x-\omega)^2}{2} \cdot \frac{1}{2\psi^2} } \right)
\]
Simplifying the above expression, we get
\begin{align*}
I
  &= \frac{\psi \sqrt{2\pi}}{2} \exp\left(-\frac{|x-\omega|}{\psi}\right)
\label{eq:solved_integral}
\end{align*}
Now substitute $I$ in equation (6)
\[
h(x) = \frac{1}{\sqrt{2\pi} \psi^2} \cdot \left[ \frac{\psi \sqrt{2\pi}}{2} \exp\left(-\frac{|x-\omega|}{\psi}\right) \right]
\]

\begin{equation}
h(x) = \frac{1}{2\psi} \exp\left(-\frac{|x-\omega|}{\psi}\right), \quad x\in \mathbb{R}
\end{equation}

This is the canonical form of the mixture PDF with location parameter $\omega$ and scale parameter $\psi >0$. Now, substitute $h(x)$ in equations (1) and (2) to get the mixture distribution based on the Arctan. The CDF and PDF, respectively are
\begin{equation}
G(x)
= 
\begin{cases}
\frac{4}{\pi} \arctan \Biggl(1-\frac{1}{2} \exp\!\left(-\dfrac{x-\omega}{\psi}\right)\Biggl), & x \geq \omega, \\[1.2em]
\frac{4}{\pi} \arctan \Biggl(\dfrac{1}{2} \exp\!\left(-\dfrac{\omega - x}{\psi}\right)\Biggl), & x < \omega,
\end{cases}
\end{equation}
and
\begin{equation}
g(x) = 
\begin{cases}
\frac{2e^{- \frac{x - \omega}{\psi}}}{\pi\psi(1+(1-\frac{1}{2}e^{- \frac{x - \omega}{\psi}})^2)}, \quad x\geq \omega, \\[1.2em]

\frac{8e^{- \frac{\omega -x}{\psi}}}{\pi\psi(4 + e^{- \frac{\omega -x}{\psi}})^2}, \quad x < \omega.
\end{cases}
\end{equation}
\section{Properties of the Proposed Model}
In this section, we derive some statistical properties of the proposed model, including the quantile function, skewness and kurtosis, moments, hazard rate function, survival function and cumulative hazard function. 
\subsection{Quantile Function}
The quantile function of a statistical distribution indicates the value of a variable at which the cumulative probability meets a given proportion. It is essential for evaluating risk and developing confidence intervals when studying the new family of distributions. 
\begin{proposition}
The quantile function ($\mathbf{u}$) associated with the proposed distribution is represented by
\begin{equation}
\mathbf{u}(p) =
\begin{cases}
\omega + \psi \ln\!\left( 2 \tan\!\left(\tfrac{\pi p}{4}\right) \right), & 0 < p < \tfrac{1}{2}, \\[1.5ex]
\omega - \psi \ln\!\left( 2 \left(1 - \tan\!\left(\tfrac{\pi p}{4}\right)\right) \right), & \tfrac{1}{2} \leq p < 1.
\end{cases}
\end{equation}
\end{proposition}
\begin{proof}
We want \(\mathbf{u}(p)\) such that \(G(x) = p\). So, for \(x \ge \omega\),  
\[
p = \frac{4}{\pi} \arctan \left( 1 - \frac{1}{2} e^{-\frac{x-\omega}{\psi}} \right),
\]  
implies
\[
e^{-\frac{x-\omega}{\psi}} = 2 \left( 1 - \tan\left(\frac{\pi p}{4}\right) \right),
\] 
\[
-\frac{x-\omega}{\psi} = \ln \left[ 2 \left( 1 - \tan\left(\frac{\pi p}{4}\right) \right) \right],
\]
Thus,  
\[
x = \omega - \psi \ln \left[ 2 \left( 1 - \tan\left(\frac{\pi p}{4}\right) \right) \right].
\]
For \(x < \omega\),  
\[
p = \frac{4}{\pi} \arctan \left( \frac{1}{2} e^{-\frac{\omega-x}{\psi}} \right),
\]
implies 
\[
-\frac{\omega-x}{\psi} = \ln \left[ 2 \tan\left(\frac{\pi p}{4}\right) \right].
\]
Thus,  
\[
x = \omega + \psi \ln \left[ 2 \tan\left(\frac{\pi p}{4}\right) \right].
\]
\end{proof}
\subsection{Skewness and Kurtosis}
Skewness indicates both the direction and degree of asymmetry in a distribution, where positive values suggest a longer tail on the right and negative values imply a longer tail on the left.
\begin{equation}
\text{Skewness} \;=\; 
\frac{\mathbf{u}_{\tfrac{1}{4}} + \mathbf{u}_{\tfrac{3}{4}} - 2\mathbf{u}_{\tfrac{1}{2}}}{\mathbf{u}_{\tfrac{3}{4}} - \mathbf{u}_{\tfrac{1}{4}}}.
\end{equation}
Kurtosis evaluates the weight of the tails and the sharpness of the peak, emphasizing the occurrence of outliers and departures from a normal distribution. 
\begin{equation}
\text{Kurtosis} \;=\; 
\frac{\mathbf{u}_{\tfrac{7}{8}} + \mathbf{u}_{\tfrac{3}{8}} - \mathbf{u}_{\tfrac{5}{8}} - \mathbf{u}_{\tfrac{1}{8}}}{\mathbf{u}_{\tfrac{3}{4}} - \mathbf{u}_{\tfrac{1}{4}}}.
\end{equation}
Here, $\mathbf{u}$ is the quantile function. When combined, these measures offer a more comprehensive insight into the shape and trustworthiness of the distribution, which is particularly crucial for financial data.
\subsection{Moments}
Moments are utilized for determining particularly statistical characteristics of distributions. In the proposed model, the $r^{th}$ moment can be obtained as 
\begin{equation}
\mu_r = 
\begin{cases}
\int_{\omega}^{\infty} x^r \frac{2e^{- \frac{x - \omega}{\psi}}}{\pi\psi(1+(1-\frac{1}{2}e^{- \frac{x - \omega}{\psi}})^2)}, \quad x\geq \omega, \\[1.2em]

\int_{-\infty}^{\omega} x^r \frac{8e^{- \frac{\omega -x}{\psi}}}{\pi\psi(4 + e^{- \frac{\omega -x}{\psi}})^2}, \quad x < \omega.
\end{cases}
\end{equation}
\subsection{Survival Function}
For the proposed distribution, survival function is defined as:  
\[
\text{sug}(x) = 1 - G(x),
\]  
\begin{equation}
\text{sug}(x) =  
\begin{cases}
1 -\frac{4}{\pi} \arctan \Biggl(1-\frac{1}{2} \exp\!\left(-\dfrac{x-\omega}{\psi}\right)\Biggl), & x \geq \omega, \\[1.2em]
1 - \frac{4}{\pi} \arctan \Biggl(\dfrac{1}{2} \exp\!\left(-\dfrac{\omega - x}{\psi}\right)\Biggl), & x < \omega,
\end{cases}
\end{equation}

\subsection{Cumulative Hazard Function}
For the proposed distribution, cumulative hazard function is defined as:  
\[
\text{cg}(x) = - \log \big( \text{sug}(x) \big),
\] 
\begin{equation}
\text{cg}(x) = 
\begin{cases}
- \log \Bigg[1 - \frac{4}{\pi} \arctan \Biggl(1-\frac{1}{2} \exp\!\left(-\dfrac{x-\omega}{\psi}\right)\Biggl)\Bigg], & x \geq \omega, \\[1.2em]
- \log \Bigg[1 - \frac{4}{\pi} \arctan \Biggl(\dfrac{1}{2} \exp\!\left(-\dfrac{\omega - x}{\psi}\right)\Biggl)\Bigg], & x < \omega,
\end{cases}
\label{eq:cumhazard}
\end{equation}
\subsection{Hazard Rate Function}
The hazard rate function is expressed as the quotient of the probability density function (PDF) and the survival function, indicating the immediate rate of failure at a specific moment. It is essential in the modelling and examination of extreme occurrences.

\begin{equation}
  Z(x; \zeta) = \frac{g(x; \zeta)}{1 - G(x; \zeta)}. 
\end{equation}

\section{Actuarial Risk Measures}
\subsection{Value at Risk(VaR)}
A mathematical risk measure ($\rho$) is a function that outputs a numerical value from $X$, a set of random variables. Using this, one can measure the risk associated with $X$. With this idea, we first select eligible members from the set $X$ of acceptable random variables and optimize $\underset{X}{\text{min}}\,  \rho (X)$ \cite{norton2021cvar, haas2009var}. Value at Risk (VaR) is a fundamental concept in risk management.  
The intuition behind the $\alpha$-percentile of a distribution is easily understood, and VaR clearly conveys its meaning: the maximum loss one can expect at a specified confidence level. The Value at Risk at confidence level $\alpha$ is computed under the proposed model as
\begin{equation}
\text{VaR}_{\alpha}(X) = \min\left\{ \vartheta \in \mathbb{R} \;\middle|\; G_X(\vartheta) \geq \alpha \right\}, \quad \frac{1}{2} < \alpha < 1,
\end{equation}
\begin{equation}
\text{VaR}_{\alpha}(X) = \omega - \psi \log(2 - 2 tan(\frac{\pi\alpha}{4})), \quad \frac{1}{2} < \alpha < 1.
\end{equation}
Figure~\ref{Fig1} shows the plot of value at risk for different values of the parameters of the proposed model.  
\begin{figure}[ht]
\begin{center}
\centerline{\includegraphics[width=0.6\columnwidth]{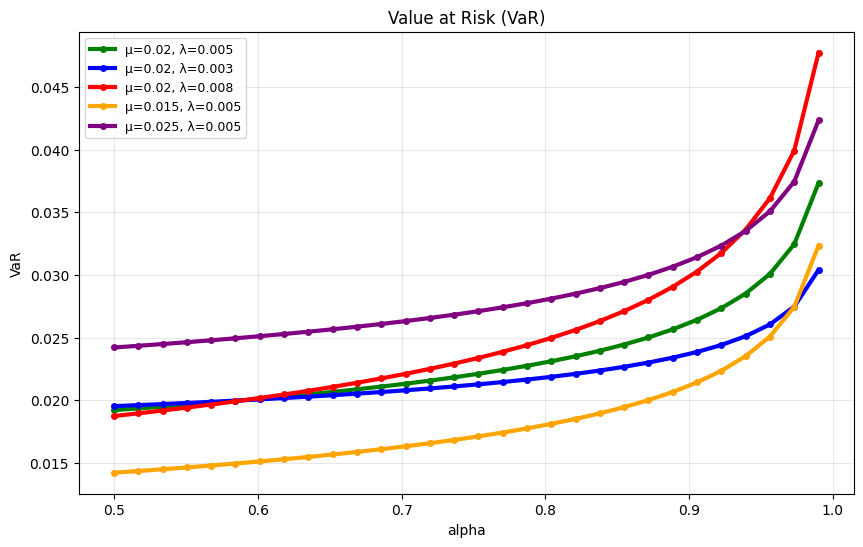}}
\caption{The effectiveness of VaR risk metrics is evaluated using the proposed model with different values of $\omega$, $\psi$}
\label{Fig1}
\end{center}
\end{figure}

\subsection{Tail Value at Risk(TVaR)}
The VaR captures the threshold of loss at a given confidence level, but it does not reflect the magnitude of losses beyond this point. The Tail Value at Risk (TVaR) \cite{acerbi2002expected} measures the expected loss given that the loss exceeds the VaR threshold. TVaR metric encompasses the expected losses found in the distribution's tail, representing the most unfavourable results that occur with a probability of $1-\alpha$. Hence, using the proposed distribution, the expression for the TVaR risk measure is  
\begin{equation}
\text{TVaR}_{\alpha}(X) = \frac{1}{1-\alpha} \int_{VaR_\alpha}^\infty x g(x) dx, \quad \frac{1}{2} < \alpha < 1.
\end{equation}
Figure~\ref{Fig2} shows the plot of tail value at risk for different values of the parameters of the proposed model.
\begin{figure}[http]
\begin{center}
\centerline{\includegraphics[width=0.6\columnwidth]{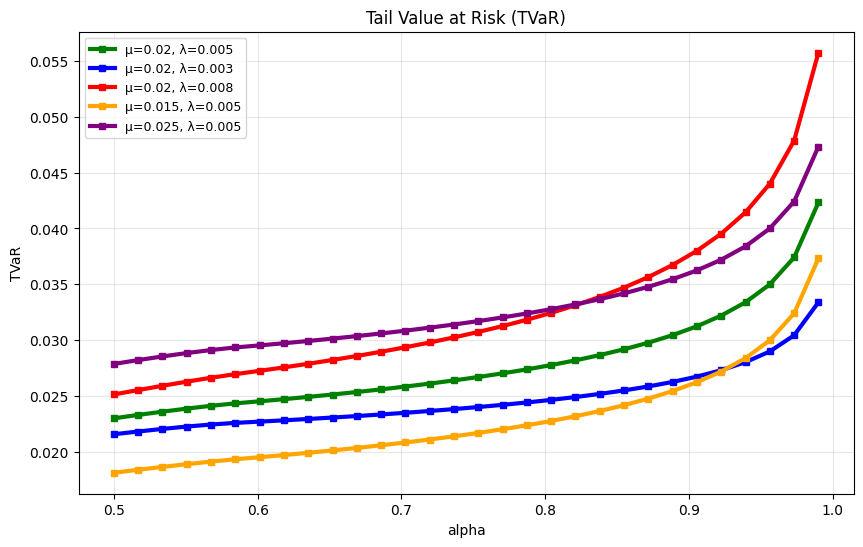}}
\caption{The effectiveness of TVaR risk metrics is evaluated using the proposed model with different values of $\omega$, $\psi$}
\label{Fig2}
\end{center}
\end{figure}
\subsection{Tail Variance(TV)}
Nonetheless, we observe that the above risk metric only reflects the average of tail loss exceeding VaR, without considering the variation of such tail loss. This is evidently insufficient for evaluating tail risk, particularly for significant losses associated with rare events. To capture the fluctuation of tail loss and thereby improve the management of extreme risk, a tail variance (TV) \cite{Furman} is proposed as a measure of right-tail variability, defined as 
\begin{align}
TV_\alpha(X) &= E(X^2 \mid X > VaR_\alpha) - (TVaR_\alpha)^2 \\
&= \frac{1}{1-\alpha} \int_{VaR_\alpha}^\infty x^2 g(x)  dx - (TVaR_\alpha)^2
\end{align}
Figure~\ref{Fig3} shows the plot of tail variance for different values of the parameters of the proposed model.
\begin{figure}[ht]
\begin{center}
\centerline{\includegraphics[width=0.6\columnwidth]{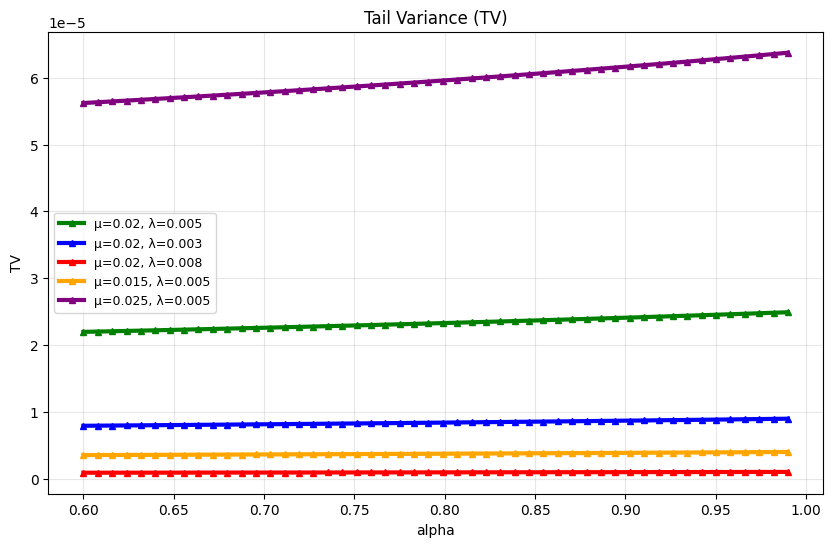}}
\caption{The effectiveness of TV risk metrics is evaluated using the proposed model with different values of $\omega$, $\psi$}
\label{Fig3}
\end{center}
\end{figure}
Table~\ref{Tab1} shows the numerical values for VaR, TVaR and TV for $\alpha \in (0.5,1)$.
\begin{table}[h]
\begin{center}
\begin{tabular}{|l|l|l|l|l|l|l|l|}
\hline
{$\alpha$} & {$\text{VaR}_{\alpha}$} & {$\text{TVaR}_{\alpha}$} & {$\text{TV}_{\alpha}$} & {$\alpha$} & {$\text{VaR}_{\alpha}$} & {$\text{TVaR}_{\alpha}$} & {$\text{TV}_{\alpha}$} \\
\hline
0.609 & 0.020187 & 0.024627 & 0.000022 & 0.827 & 0.023650 & 0.028356 & 0.000023 \\
0.663 & 0.020799 & 0.025296 & 0.000022 & 0.881 & 0.025357 & 0.030146 & 0.000024 \\
0.718 & 0.021535 & 0.026095 & 0.000023 & 0.936 & 0.028229 & 0.033109 & 0.000024 \\
0.772 & 0.022451 & 0.027080 & 0.000023 & 0.990 & 0.037341 & 0.042322 & 0.000025 \\
\hline
\end{tabular}
\end{center}
\caption{\label{fontsizes} Estimated VaR, TVaR and TV at different confidence levels and for fixed $\omega = 0.02$, $\psi = 0.005$, based on the proposed model.}
\label{Tab1}
\end{table}

\section{Maximum Likelihood Estimation (MLE)}
Estimating parameters through the maximum likelihood approach enables the use of all the information present in the data to achieve accurate estimates. Moreover, it offers flexibility and is suitable for a variety of statistical models, making it easier to analyze intricate distributions. For a sample $x_1, \dots, x_n$, the likelihood function is
\begin{equation}
L(\omega, \psi) = \prod_{i=1}^n g(x_i).
\end{equation}
Then the Log-likelihood function for the proposed distribution is obtained as




 



\begin{align*}
\ell(\omega, \psi) = &\; n_1 \ln 2 + n_2 \ln 8 - \frac{1}{\psi} \left[ \sum_{x_i \ge \omega} (x_i - \omega) + \sum_{x_i < \omega} (\omega - x_i) \right] - n \ln \pi - n \ln \lambda \\
& - \sum_{x_i \ge \omega} \ln\left[ 1 + \left( 1 - \frac12 e^{-(x_i - \omega)/\psi} \right)^2 \right] - 2 \sum_{x_i < \omega} \ln\left( 4 + e^{-(\omega - x_i)/\psi} \right).
\end{align*}.

\section{Insurance Data Analysis}
For the applicability of the proposed mixture model, we compare it with the baseline Gaussian, Rayleigh, and NR distributions for modeling the heavy-tailed insurance dataset. We examine the information criteria such as AIC, CAIC and HQIC. The models with the lowest information values are considered the best. Here are illustrations of information criteria for reference.
\begin{align*}
\text{AIC}  &= -2l + 2r, \\
\text{CAIC} &= -2l + \frac{2mr}{m - r - 1}, \\
\text{HQIC} &= -2l + 2r \ln(\ln(m)),
\end{align*}
here, $l$ signifies the log-likelihood function evaluated at the maximum likelihood estimates (MLEs), $m$ corresponds to the sample size, and $r$ specifies the count of parameters within the statistical model. We use the Insurance dataset, which includes the unemployment insurance data from July 2008 to April 2013  from the Department of Labour, Licensing, and Regulation of the State of Maryland, USA (https://catalog.data.gov/dataset/). This dataset consists of a total of 58 observations and 20 attributes, among which we focus on the attribute titled "Number of First UI Checks Issued - Ex. Fed Employees”. This attribute measures the monthly data of unemployment insurance payments distributed to former federal government employees.
Table~\ref{Tab2} displays insights on insurance data, while Figures~\ref{Fig4} and ~\ref{Fig5} illustrate the histogram and box plot of the dataset along with descriptive statistics.
\begin{table}[h!]
\centering
\begin{tabular}{|cccccccccc|}
\hline
0.052 & 0.033 & 0.039 & 0.050 & 0.029 & 0.052 & 0.060 & 0.032 & 0.057 & 0.064 \\
0.061 & 0.064 & 0.041 & 0.036 & 0.050 & 0.053 & 0.061 & 0.068 & 0.060 & 0.050 \\
0.064 & 0.057 & 0.061 & 0.059 & 0.069 & 0.070 & 0.137 & 0.170 & 0.100 & 0.090 \\
0.222 & 0.109 & 0.068 & 0.063 & 0.056 & 0.090 & 0.074 & 0.095 & 0.114 & 0.133 \\
0.066 & 0.075 & 0.072 & 0.054 & 0.057 & 0.052 & 0.066 & 0.069 & 0.083 & 0.044 \\
0.060 & 0.080 & 0.058 & 0.080 & 0.080 & 0.052 & 0.065 & 0.073 &        &       \\
\hline
\end{tabular}
\caption{Insurance dataset values}\label{Tab2}
\end{table}

\begin{figure}[h!]
    \centering
    \begin{minipage}{0.48\textwidth}
        \centering
        \includegraphics[width=\linewidth]{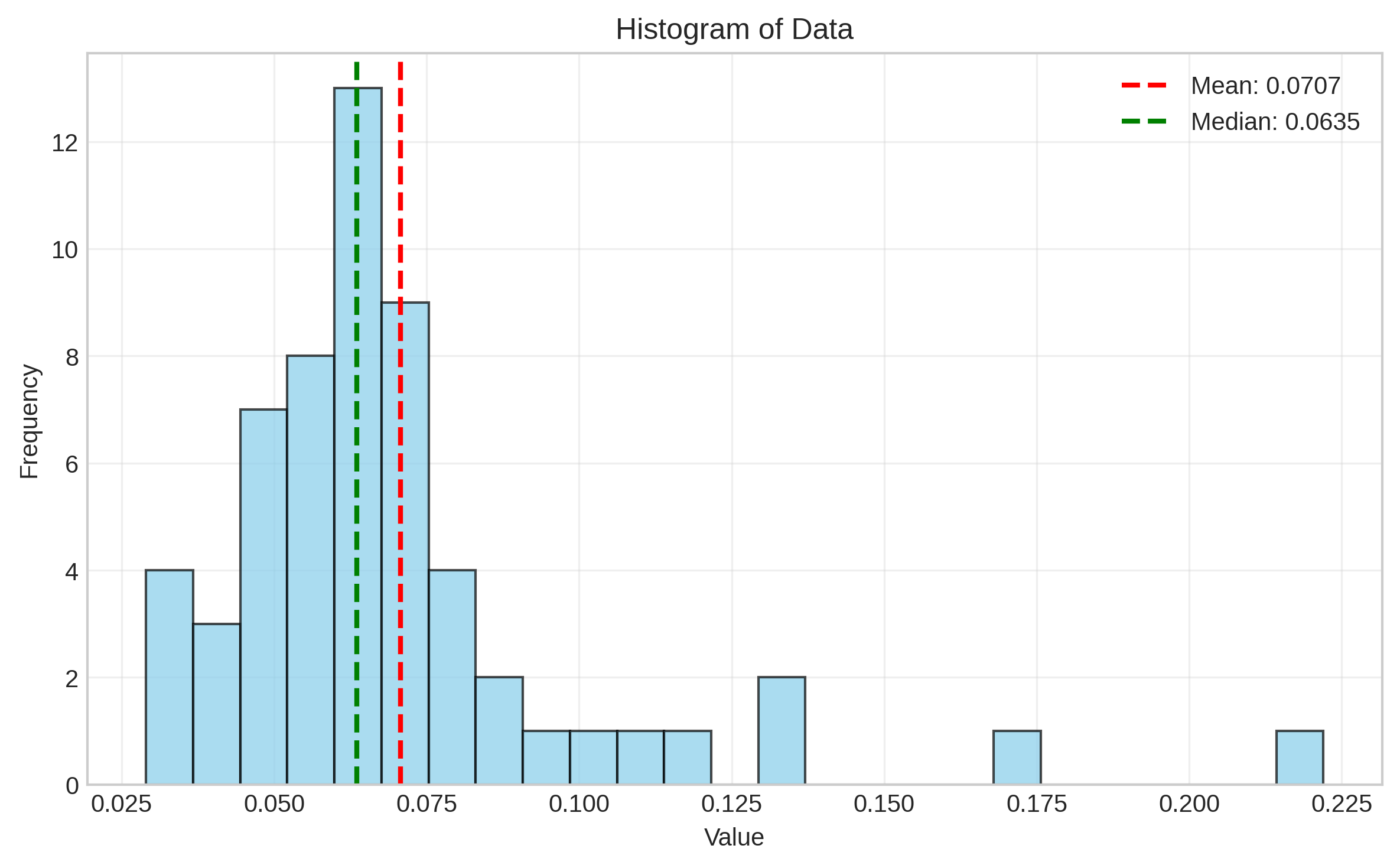}
        \caption{Histogram of Insurance Dataset} \label{Fig4}
    \end{minipage}
    \hfill
    \begin{minipage}{0.48\textwidth}
        \centering
        \includegraphics[width=\linewidth]{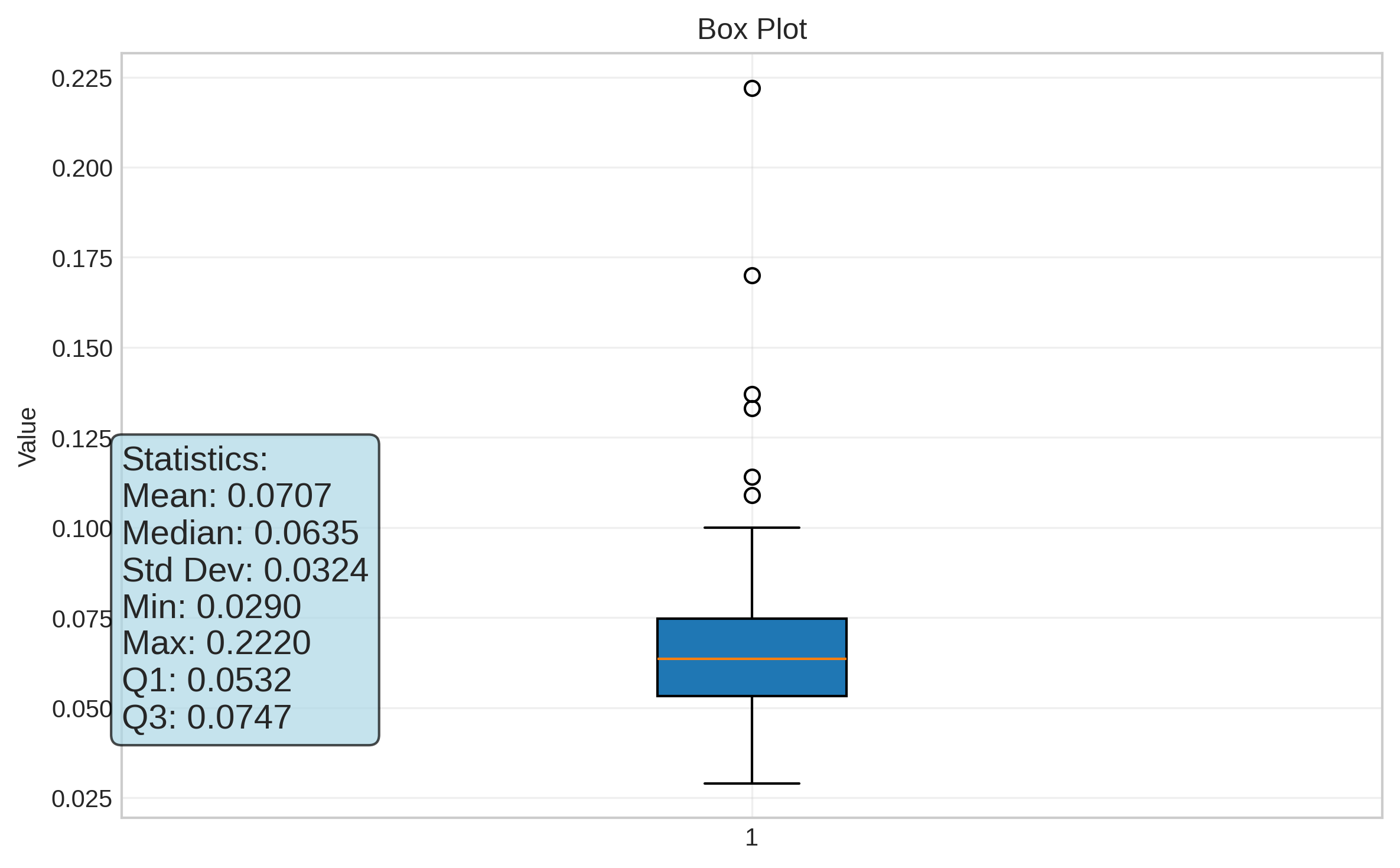}
        \caption{Boxplot of Insurance Dataset}\label{Fig5}
    \end{minipage}
\end{figure}
To compare the density of distributions with the arctan mixture model, they are plotted in Figure~\ref{Fig6}, showing that the arctan mixture model exhibits fatter tails.
\begin{figure}[ht]
\begin{center}
\centerline{\includegraphics[width=0.6\columnwidth]{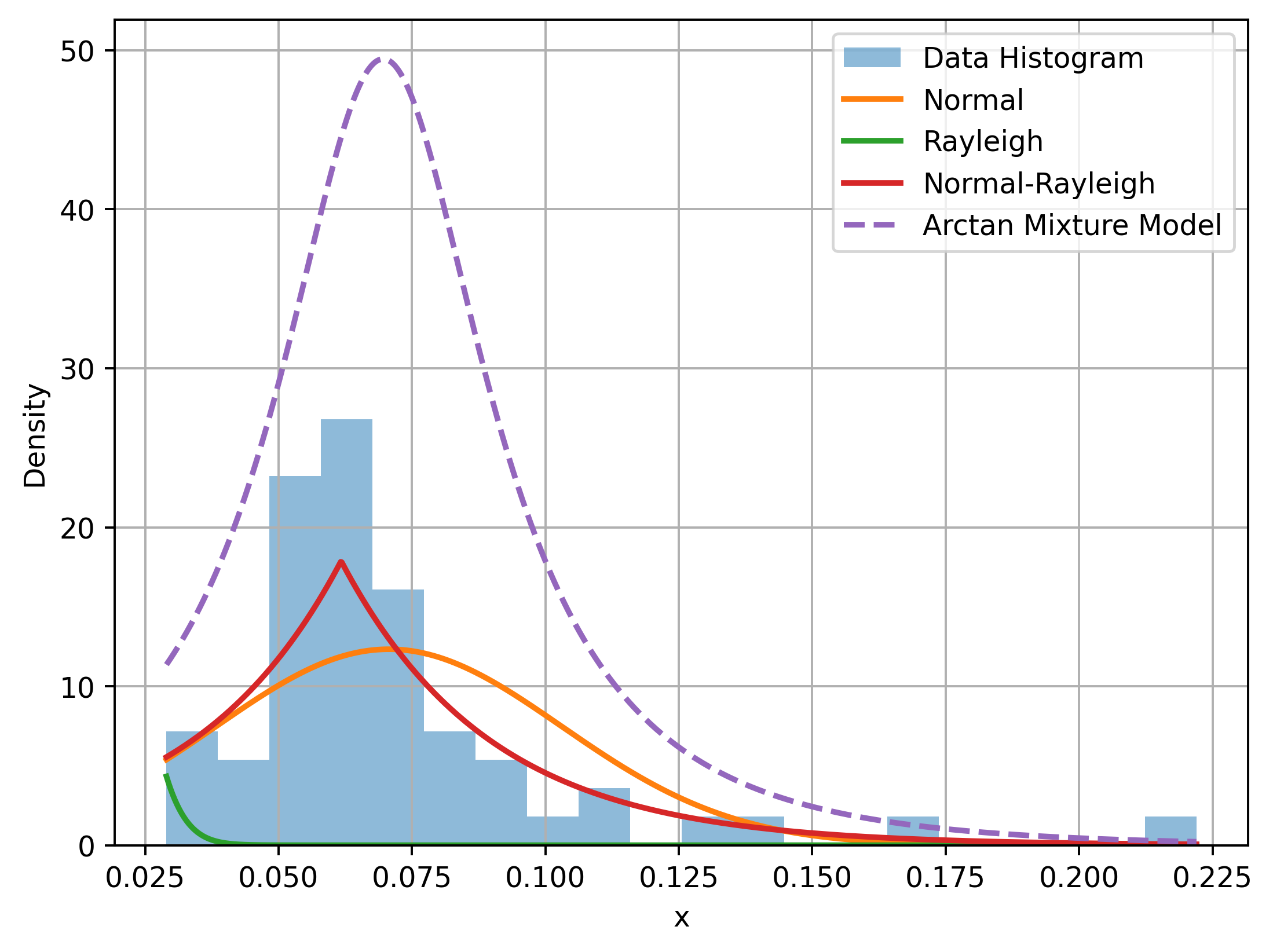}}
\caption{Comparison of different densities showing that heavy fat tails for the arctan mixture density.}
\label{Fig6}
\end{center}
\end{figure}
Table~\ref{Tab3} provides the information criterion values for the Insurance Dataset, comparing the proposed model with baseline models.
\begin{table}[h!]
\centering
\begin{tabular}{|l| l| l| l| l| l| l| l|}
\hline
\textbf{Model} & \textbf{Par.} & $r$ & \textbf{LL} & \textbf{AIC} & \textbf{BIC} & \textbf{CAIC} & \textbf{HQIC} \\
\hline
\textbf{Arctan-GR} & $\omega=0.0901 $ & 2 & \textbf{188.6452} & \textbf{-371.2904} & \textbf{-365.1091} & \textbf{-370.8460} & \textbf{-368.8827} \\
                 (Mixture)   & $\psi=0.0162$ &&&&&& \\
\hline                    
Gaussian        & $\omega=0.0707, $                & 2 & 116.6856 & -229.3711 & -225.2502 & -229.1529 & -227.7660 \\
& $\eta=0.0324$ &&&&&& \\
\hline
Rayleigh      & $\eta=0.0229$                               & 1 & -54.5752 & 111.1505  & 113.2109  & 111.2219  & 111.9531  \\
&  &&&&&& \\
\hline
Gaussian-Rayleigh      & $\omega=0.0707$                     & 2 & 126.2571 & -248.5142 & -244.3933 & -248.2960 & -246.9090 \\
 (NR) & $\psi=0.0209$ &&&&&& \\
\hline
\end{tabular}
\caption{\label{Tab3} Model Comparison Table}
\vspace{0.5em}
\begin{minipage}{0.95\textwidth}
\footnotesize
\textbf{Note:} For AIC, BIC, CAIC, and HQIC, \textit{lower values indicate better fit}. \\
For LL (Log-Likelihood), \textit{higher values indicate better fit}. \\
$r =$ number of parameters in the model. \\
\textbf{Bold} shows a better fit model among others.
\end{minipage}
\end{table}
Table~\ref{Tab3} shows that the arctan-based mixture model performs better in fitting the heavy-tailed data than other baseline models. 

\begin{table}[ht!]
\begin{center}
\begin{tabular}{|c|c|c|c|}
\hline
\textbf{Confidence Level ($\alpha$)} & $\text{VaR}_{\alpha}$ & $\text{TVaR}_{\alpha}$ & $\text{TV}_{\alpha}$ \\
\hline
75\%	& 0.073383	& 0.094426	& 0.000478 \\
80\%	& 0.077809	& 0.099158	& 0.000486 \\	
85\%	& 0.083671	& 0.105353	& 0.000494 \\	
90\%	& 0.092182	& 0.114226	& 0.000503 \\	
95\%	& 0.107228	& 0.129665	& 0.000512 \\	
99\%	& 0.143366	& 0.166143	& 0.000520 \\	
\hline
\end{tabular}
\end{center}
\caption{\label{Tab4} Empirical VaR, TVaR and TV at different confidence levels for insurance data.}
\end{table}
Table~\ref{Tab4} results show the empirical VaR, TVaR and  TV risk measures for the insurance dataset.

\clearpage
\section{Conclusions}

In this article, we introduce a new modeling framework designed to address the challenges posed by heavy-tailed financial datasets and applications in actuarial risk management. The proposed model demonstrates superior fitting performance compared to current baseline distributions, showcasing its practical value in financial risk evaluation. The combination of Gaussian and Rayleigh distributions through arctangent transformation is particularly successful in representing extreme market behaviour and tail dependencies that are frequently found in real-world financial data. We analyze statistical characteristics of the proposed distribution, validating its applicability in the from of a financial perspective for insurance dataset and additionally actuarial risk measures are also calculated for the proposed model. The maximum likelihood method (MLE) has been utilized to estimate the parameters of the proposed distribution.  

Future research could build on this work by investigating additional statistical properties of the proposed distribution. Additionally, a thorough assessment of the model's effectiveness across a variety of financial and economic datasets from different market conditions and regions would enhance its applicability and practical use in a range of risk management scenarios. It is also possible to investigate the new mixture type of distribution based on the arctan function. The arctan-based mixture distribution can be applied to real datasets in healthcare, reliability, and other areas, allowing for a comparison of its performance against other well-known distributions.

\end{document}